\documentclass[journal]{IEEEtran}

\newcommand{\qed}{\rule{2mm}{2mm}\bigskip}
\usepackage{graphicx}
\usepackage{epstopdf}
\usepackage{stmaryrd}
\usepackage{amssymb,amsmath}
\usepackage{multirow,url}
\usepackage[usenames]{color}
\usepackage{cite,url}
\usepackage{setspace}
\newtheorem{algorithm}{\bf Algorithm}
\newtheorem{theorem}{\bf Theorem}

\newtheorem{lemma}{\bf Lemma}
\newtheorem{definition}{\bf Definition}
\newtheorem{proposition}{\bf Proposition}

\newtheorem{remark}{\bf Remark}

\newtheorem{example}{\bf Example}
\DeclareMathOperator*{\argmax}{argmax}

\def\qed{\quad \vrule height7.5pt width5.17pt depth0pt}

\ifodd 0
\newcommand{\rev}[1]{{\color{blue}#1}} 
\newcommand{\com}[1]{\textbf{\color{red}(COMMENT: #1)}} 
\newcommand{\clar}[1]{\textbf{\color{green}(NEED CLARIFICATION: #1)}}

\else
\newcommand{\rev}[1]{#1}
\newcommand{\com}[1]{}
\newcommand{\clar}[1]{}
\fi

\begin{document} 

\title{Spectrum Sharing as Spatial Congestion Games}
\author{Sahand Ahmad, Cem Tekin, Mingyan Liu, Richard Southwell, Jianwei Huang 
\thanks{This work is supported by NSF award CNS-0238035, CCF-0910765, and through collaborative participation in the Communications and Networks Consortium sponsored by the U. S. Army Research Laboratory under the Collaborative Technology Alliance Program Cooperative Agreement DAAD19-01-2-0011. An earlier version of this paper appeared in GameNets'09.}
\thanks{S. Ahmad, C. Tekin and M. Liu are with the Electrical Engineering and Computer Science Department, University of Michigan, Ann Arbor, MI 48105, USA,
        {\{shajiali,cmtkn,mingyan\}@eecs.umich.edu},
R. Southwell and J. Huang are with the Chinese University of Hong Kong, Hong Kong,
	richardsouthwell254@gmail.com, jwhuang@ie.cuhk.edu.hk.}
}

\maketitle

\begin{abstract}

In this paper, we \rev{present} and analyze the properties of a new class of games - the spatial congestion game (SCG), which is a generalization of the classical congestion game (CG). In a classical congestion game,
multiple users share the same set of resources and a user's payoff for using any resource is a function of the total number of users sharing it.  \rev{As a potential game}, this game enjoys some very appealing properties, including the existence of a pure strategy Nash equilibrium (NE) and that every improvement path is finite and leads to such a NE (also called the finite improvement property or FIP). 
\rev{While it's tempting to use this model to study spectrum sharing, it does not capture the spatial reuse feature of wireless communication,} 
where resources (interpreted as channels) may be {\em reused} without increasing congestion provided that users are located far away from each other.  This motivates \rev{us} to study an extended form of the congestion game where a user's payoff for using a \rev{resource} is a function of the number of its {\em interfering} users sharing \rev{it}. 
This naturally \rev{results in} a {\em spatial congestion game} \rev{(SCG)}, \rev{where} users are placed over a network (or a conflict graph).  We study fundamental properties of a spatial congestion game; in particular, we seek to answer under what conditions this game possesses the \rev{finite improvement property} or a \rev{Nash equilibrium}.  We also discuss the implications of these results when applied to wireless spectrum sharing. 

\end{abstract}



\section{Introduction}\label{sec-intro}




In this paper, we study \rev{a} \emph{spatial congestion game} (SCG), which is a generalized form of the class of non-coopertive strategic games known as {congestion games} (CG) \cite{Rosenthal1973,vocking2006cgo}. We analyze the properties of the SCG and discuss its application to spectrum sharing in multi-channel wireless networks.


In a classical congestion game, multiple users share multiple resources. A user's payoff \footnote{\rev{One can also consider the cost of using a resource instead of payoff}.  If we define the cost as the inverse of the payoff, then maximizing the payoff is equivalent to minimizing the cost. For simplicity of presentation, we will only refer to the maximization of payoff in this paper.}
for using a particular resource depends on the number of users \rev{simultaneously} using that resource. A formal description is provided in Section \ref{sec-review}.
%
%
%
The congestion game framework is well suited to model resource competition
where the resulting payoff is a function of the level of congestion (number of active users).
It has been extensively studied within the context of \rev{wireline} network routing, see for instance the congestion game studied in \cite{Fabrikant04}, 
\rev{where} each source node seeks the minimum delay path to a destination node, and the delay of a link depends on the number of flows going through that link.
\rev{It has recently been used in wireless network modeling, e.g., access point selection in WiFi networks \cite{ercetin2008association,chen2010distributed}, resource competition in multicamera wireless surveillance networks \cite{shiang2010information}, uplink resource allocation in multichannel wireless access networks \cite{altman2009potential},  wireless channels with multipacket reception capability \cite{sanyal2010congestion}, and the impact of interference set in studying the congestion game in wireless mesh networks \cite{argento2010access}. } 

A congestion game enjoys many nice properties: \rev{it has a} pure strategy Nash Equilibrium (NE), and any asynchronous improvement path is finite and will lead to a pure strategy NE.  The latter property is also called the \emph{finite improvement property} (FIP): 
%
\rev{local greedy updates of selfish users collectively optimize a global objective known as the potential function, and such updates converge in a finite number of steps regardless of the updating sequence.} 


Due to the above reasons, it is tempting to model resource competition in a wireless communication system as a congestion game.  However, the standard congestion game fails to capture a critical aspect of resource sharing in wireless communication: {\em interference}.
A key assumption underlying the congestion game model is that all users have an equal impact on the congestion, and therefore all that matters is the total number of users of a resource. 
This however is not true in wireless communication.
Specifically, if we consider channels as resources, then sharing the same channel is complicated by  interference; a user's payoff (e.g., channel quality, achievable rates, etc.) depends on {\em who} the other users are and how much interference it receives from them.  If all other simultaneous users are sufficiently far away, then sharing may not cause any performance degradation, a feature commonly known as \emph{spatial reuse}.
%

The above consideration poses significant challenge in using the congestion game model depending on what type of user objectives we are interested in.  In our recent work \cite{liu-allerton08}, we addressed the user-specific interference issue within the traditional congestion game framework, by introducing a  concept called {\em resource expansion}, where we define virtual resources as certain spectral-spatial unit that allows us to capture pair-wise interference.  This approach was shown to be quite effective for user objectives like interference minimization.  
%

In this paper, we take a different  and more general approach, where we generalize the standard congestion games to directly account for the interference relationship and spatial reuse in wireless networks.  This class of generalized games will be referred to as {\em spatial congestion games} (SCG).  \rev{A key ingredient in this generalization is an interference graph describing the congestion relationship among users.} 
%
%
%
In using a resource (a wireless channel), a user's payoff is a function of the total number of users who are using the same resource and are {\em within its interference set} (i.e., connected to it by edges).  Therefore, resources are {\em reusable} beyond a user's interference set.  \rev{The original congestion game is now a special case of the extended SCG when the underlying interference graph is complete (i.e., every user interferes with every other user) \footnote{In our preliminary work \cite{liu-gamenets09} we used the term {\em network congestion games}. However, to better differentiate this class of games from routing games (see e.g., \cite{TR02,CMNV05}) which are also sometimes referred to as network congestion games, we will use the term {\em spatial congestion games} in this paper.  Note that a routing game is essentially a classical congestion game in which a user's strategy space consists of a set of feasible routes and each route consists of multiple resources (links).}.} 

\rev{Congestion games played on networks have been studied before in \cite{Bilo2008}}, where each user has the same linear payoff function. Our SCG model allows user-specific payoff functions of more general forms. In this sense our model is also a generalization of that considered in \cite{Bilo2008}.  This allows us to model systems like cognitive radio networks where \rev{technologies may vary from user to user.} 

The applicability of the SCG to a multi-channel, multi-user wireless communication system can be easily understood.  Specifically, we consider a system where a user can only access one channel at a time, but can switch between channels.  A user's principal interest lies in optimizing its own performance (e.g., its data rate) by selecting the best channel for itself.  This and similar problems have recently captured increasing interest from the research community, particularly in the context of cognitive radio networks (CRN) and software defined ratio (SDR) technologies, where devices expected to have far greater flexibility in sensing channel availability and moving their operating frequencies. \rev{More broadly, the SCG framework} is potentially applicable to many other scenarios where resources are shared over space.

In subsequent sections we will examine what properties a SCG has. Our main findings are summarized as follows for \emph{undirected} network graphs and \emph{non-increasing} payoff functions (in the number of users sharing a resource):

\begin{enumerate}
\item The FIP property is preserved in an SCG with only two resources/channels.  Counter examples exist for three or more resources.
\item The FIP property is preserved in an SCG when all resources are identical to a user \rev{(but may be different to different users)}.  In the context of multi-channel communications, this means each channel is of equal bandwidth and quality for a user. 
\item A pure strategy NE exists in an SCG over a tree network, a loop, a regular bipartite network, and when there is a dominating resource.
\item We identify counter examples to show that an NE does not necessarily exist when the network graph is directed (meaning that the interference relationship between users is asymmetric), or when users' payoff functions are non-monotonic. 
\end{enumerate}


It should be mentioned that game theoretic approaches have often been used to devise effective decentralized solutions to a multi-agent system.   Within the context of wireless communication networks and interference modeling, different classes of games have been studied.  An example is the well-known {\em Gaussian interference game} \cite{Tse2007,Cioffi2002}, in which a player can spread a fixed amount of power arbitrarily across a continuous bandwidth, and tries to maximize its total rate in a Gaussian interference channel over all possible power allocation strategies.  
\rev{The Bayesian form of the Gaussian interference game was studied in \cite{Goldsmith} in the case of incomplete information.  In addition, a market based power control mechanism was investigated via supermodularity in \cite{Huang2006}, and using externality in \cite{shruti2008}.  A spectrum sharing similar to the one studied here was investigated in \cite{ali2010} using a mechanism design approach in seeking a globally optimal solution.} 
In our problem the total power of a user is not divisible, and it can only use it in one channel at a time.  This setup is more appropriate for scenarios where the channels have been pre-defined, and the users do not have the ability to access multiple channels simultaneously (which is the case with many existing devices).

The organization of the remainder of this paper is as follows.  In Section \ref{sec-review} we present a brief review on the background of the classical CG, and formally define the class of SCG in Section \ref{sec-problem}.  We then derive conditions under which SCG possesses the finite improvement property in Section \ref{sec-fip}.  We further show a series of conditions, on the underlying network graph and on the user payoff function in Section \ref{sec-graph}, under which an SCG has a pure strategy NE.  We discuss extensions to our work in Section \ref{sec-discussion} and conclude the paper in \ref{sec-conclusion}.


\section{A Review of Congestion Games}
\label{sec-review}

In this section we provide a brief review on the definition of congestion games and their known properties\footnote{This review along with some of our notations are primarily based on references \cite{Rosenthal1973,vocking2006cgo,Monderer1996}.}.
We then discuss why the classical congestion game does not capture spatial reuse and motivate our generalized spatial congestion games.

\subsection{Congestion Games}



Congestion games \cite{Rosenthal1973,vocking2006cgo} are a class of strategic games given by the tuple $({\cal I}, {\cal R}, (\Sigma_i)_{i\in{\cal I}}, (g_r)_{r\in{\cal{R}}})$, where ${\cal I}=\{1, 2, \cdots, N\}$ denotes a set of users, ${\cal R}=\{1, 2, \cdots, R\}$ a set of resources, $\Sigma_i\subset 2^{\cal R}$ the strategy space of player $i$, and $g_r: \mathbb{N}\rightarrow\mathbb{Z}$ a payoff (or cost) function associated with resource $r$.
The payoff (cost) $g_r(\cdot)$ of resource $r$ is a function of the total number of users using that resource, and in general is assumed to be non-increasing (non-decreasing).  A player in this game aims to maximize (minimize) its total payoff (cost) which is the sum total of payoff (cost) over all resources its strategy involves. For the rest of the paper, we will only refer to payoff maximization.

Denoting by $\boldsymbol{\sigma}=(\sigma_1, \sigma_2, \cdots, \sigma_N)$ the strategy profile, where $\sigma_i\in \Sigma_i$, user $i$'s total payoff is given by
\begin{eqnarray}
q^i(\boldsymbol{\sigma}) = \sum_{r\in \sigma_i} g_r(n_r(\boldsymbol{\sigma})) ~, 
\end{eqnarray}
where $n_r(\boldsymbol{\sigma})$ is the total number of users using resource $r$ under the strategy profile $\boldsymbol{\sigma}$, and $r\in \sigma_i$ denoting that user $i$ selects resource $r$ under $\boldsymbol{\sigma}$.

We can define Rosenthal's potential function $\phi: \Sigma_1 \times \Sigma_2 \times \cdots \times \Sigma_N \rightarrow \mathbb{Z}$ as
\begin{eqnarray}
\phi(\boldsymbol{\sigma}) = \sum_{r\in{\cal R}} \sum_{i=1}^{n_r(\boldsymbol{\sigma})} g_r(i)
= \sum_{i=1}^{N} \sum_{r\in\sigma_i} g_r(m_r^{i}(\boldsymbol{\sigma})) ~, \label{eqn-change-of-sums}
\end{eqnarray}
where the second equality comes from exchanging the two sums, and $m_r^{i} (\boldsymbol{\sigma})$ denotes the number of players who use resource $r$ under strategy  $\boldsymbol{\sigma}$ and whose corresponding indices do not exceed $i$ (i.e., in the set $\{1, 2, \cdots, i\}$).

Next we show that the change in user $i$'s payoff as a result of its unilateral move (i.e., all other users' strategy $\sigma_{-i}$ remain fixed) is exactly the same as the change in the potential function. This implies that the potential function may be viewed as a global objective function.
Consider player $i$, who unilaterally moves from strategy $\sigma_i$ (within the profile $\boldsymbol{\sigma}=(\sigma_{i},\sigma_{-i})$) to strategy $\sigma_i^{\prime}$ (within the profile $\boldsymbol{\sigma}^{\prime}=(\sigma_{i}^{\prime},\sigma_{-i})$).  The change of potential function is
\begin{eqnarray*}
&& \phi (\sigma_{i}^{\prime},\sigma_{-i})-\phi (\sigma_{i},\sigma_{-i}) \\
&=& \sum_{r\in\sigma_i^{'}, r\not\in\sigma_i} g_r(n_r(\boldsymbol{\sigma})+1) -
\sum_{r\in\sigma_i, r\not\in\sigma_i^{'}} g_r(n_r(\boldsymbol{\sigma})) \nonumber \\
&=& \sum_{r\in\sigma_i^{'}} g_r(n_r(\boldsymbol{\sigma}^{'})) - \sum_{r\in\sigma_i} g_r(n_r(\boldsymbol{\sigma}))\\
&=& g^i(\sigma_{-i}, \sigma_i^{'}) - g^{i}(\sigma_{-i}, \sigma_i) ~.
\end{eqnarray*}
The second equality comes from the fact that the number of total users does not change for any \rev{resource} that is used by both strategies $\sigma_i$ and $\sigma_i^{'}$.  To see why the first equality is true, set $i=N$, in which case this equality is a direct consequence of equation (\ref{eqn-change-of-sums}).  To see why this is true for any $1\leq i\leq N$, simply note that the ordering of users is arbitrary so any user making a change may be viewed as the $N$th user.  

Consider now a sequence of strategy changes made by users asynchronously,  in which each change improves the corresponding user's payoff (this is referred to as a sequence of improvement steps).  The result in the previous paragraph shows that the potential function also improves in every such change sequence.  Since the potential function of any strategy profile is finite, we have the following result  \cite{vocking2006cgo}: 
\begin{proposition}[finite improvement property (FIP)]
For every congestion game, every sequence of \rev{asynchronous} improvement steps is finite and converges to a pure strategy Nash Equilibrium (NE). Furthermore, this NE is a local optimum of the potential function $\phi$, defined as a strategy profile where changing one coordinate cannot result in a greater value of $\phi$.
\end{proposition}

It is not difficult to see why the standard definition of a congestion game does not capture spatial reuse of wireless communication.  In particular, if we consider channels as resources, then the payoff $g_r(n)$ for using channel $r$ when there are $n$ simultaneous users does not reflect reality:
the function $g_r(\cdot)$ in general takes a user-specific argument since different users experience different levels of interference even when using the same resource.
This user specificity is also different from that studied in \cite{milchtaich1996cgp}, where $g_r(\cdot)$ is a user-specific function $g_r^{i}(\cdot)$ but it takes the \emph{same} non-user specific argument $n$.
%
To analyze and understand the consequence of this difference,
we would need to extend and generalize the definition of the standard congestion game.

For the rest of this paper, the term {\em player} or {\em user} specifically refers to a {\em pair} of transmitter and receiver in a wireless network.  Interference in this context is between one user's transmitter and another user's receiver.  This is commonly done in the literature, see for instance \cite{Tse2007}.  We will also assume that each player has a fixed transmit power.
%


\section{Problem Formulation}
\label{sec-problem}


In this section we formally define our generalized congestion game, the {\em spatial congestion game} (SCG).
%
Specifically, an $N$-\rev{player} SCG is given by $\Gamma_{N}=({\cal I}, {\cal R}, (\Sigma_i)_{i\in\cal{I}}, \{\mathcal{K}_i\}_{i\in\cal{I}}, \{g_r^i\}_{r\in{\cal{R}}, i\in\cal{I}})$, where $\mathcal{K}_i$ is the interference set of user/player $i$ (i.e., users interfering with user $i$), while all other elements maintain the same meaning as in a standard CG. The payoff user $i$ receives for using resource $r$ is given by $g_r^i(n_r^i(\boldsymbol{\sigma})+1)$ where $n_r^i(\boldsymbol{\sigma}) = |\{j: r \in \sigma_j , j \in \mathcal{K}_i\}|$.  That is, user $i$'s payoff for using resource $r$ is a (user-specific) function of the number of users interfering with itself, plus itself.  Here we have explicitly made the payoff functions user-specific, as evidenced by the index $i$ in $g_r^i(\cdot)$.  This is done in an attempt to capture the fact that users with different coding/modulation schemes may obtain different rates from using the same channel even when facing the same level of interferences. 


%

A user's payoff is the sum of payoffs from all the resources it uses.
Note that if a user is allowed to simultaneously use all available resources, then its best strategy is to simply use all of them regardless of other users, provided that $g_r^i$ is a non-negative function.  If all users are allowed such a strategy, then the existence of an NE is trivially true.

In this paper, we will limit our attention to the case where each user is allowed only one channel at a time, i.e., its strategy space $\Sigma_{i} =  \mathcal{R}$ consists of $R$ single channel strategies.
In this case, the payoff user $i$ receives for using a single channel $r$ is given by $g_r^i(n_r^i+1)$ where $n_r^i(\boldsymbol{\sigma})=|\{j: r = \sigma_j , j \in \mathcal{K}_i \}|$.

It is easy to see that we can equivalently represent this problem on the following directed graph, where a node represents a user and a directed edge connects node $i$ to node $j$ if and only if $i \in \mathcal{K}_j$.
The spatial congestion game can now be stated as a coloring problem \footnote{\rev{We will use several colored graphs in our analysis, which may not show as effectively in a black/white version.}},  
where each node picks a color and receives a value depending on the conflict (number of same-colored neighbors to a node); the goal is to see whether an NE exists and whether a decentralized selfish scheme leads to an NE.  In this paper we will limit our attention to the case of undirected graphs, where there is an undirected edge between nodes $i$ and $j$ if and only if $i \in \mathcal{K}_j$ and $j \in \mathcal{K}_i$.  This has the intuitive meaning that if node $i$ interferes with node $j$, the reverse is also true.  This symmetry does not always hold in reality, but is \rev{often a good approximation}, and helps us obtain meaningful insight. Another reason for this assumption is that an NE does not always exist in a directed graph (as we show in the Appendix via a counter example). 

For simplicity of exposition, in subsequent sections we will often present the problem in its coloring version, and will use the terms {\em resource}, {\em channel}, {\em color}, and {\em strategy} interchangeably.
\rev{For the remainder of the paper, unless stated otherwise} we shall assume that every SCG we consider has the following properties: 
%
%
(1) users only employ one resource at a given time; 
%
(2) the payoff functions are user-specific and non-increasing; 
%
and (3) the interference graph is undirected.


\section{Existence of the Finite Improvement Property}
\label{sec-fip}


In this section we investigate whether the SCG possesses the FIP property as in the traditional CG.  \rev{If a game has this property}, it immediately follows that it has an NE as we described in Section \ref{sec-review}.
Below we show that in the following three cases an SCG possesses the FIP property: (1) when there are only two resources to choose from, (2) when all resources are identical to a user, for all users, and (3) the graph is complete.  

%

\subsection{The Finite Improvement Property for 2 Resources}

We establish this result by contradiction. Suppose that we have a sequence of asynchronous\footnote{We will remove the word {\em asynchronous} in subsequent presentation with the understanding that whenever we refer to updates they are assumed to be asynchronous updates, i.e., there will not be two or more users changing their strategies simultaneously at any time.} updates that starts and ends in the exact same \emph{state} (e.g., color assignment) for all users. We denote such a sequence by
\begin{eqnarray}
U = \{u(1), u(2), \cdots u(T)\},
\end{eqnarray}
where $u(t)\in \{1, 2, \cdots, N\}$ denotes the user making the change at time $t$, and $T$ is the length of this sequence.
The starting state  of the system is given by
\begin{eqnarray}
S(1) = \{s_1(1), s_2(1), \cdots, s_N(1)\},
\end{eqnarray}
where $s_i(1) \in \{ r, b\}$, i.e., the color of each user is either ``r'' for Red, or ``b'' for Blue.  A user $i$'s color $s_{i}(t)$ is defined for time $t^{-}$, i.e., right before a color change is made by some user at time $t$.
Since there are only two colors, we use the notation $\bar{s}$ to denote the opposite color of a color $s$.  


Since this sequence of updates form a loop, 
we can naturally view \rev{them} as being placed \rev{on} a circle, starting at time $1^{-}$ and ending at $T^{+}$, when the system returns to its original state.  This is shown in Figure \ref{fig-circle}. Note that traversing the circle starting from any point results in an improvement path; hence the notion of a starting point becomes inconsequential. 
\begin{figure}[h]
\centering
\includegraphics[width=2.0in]{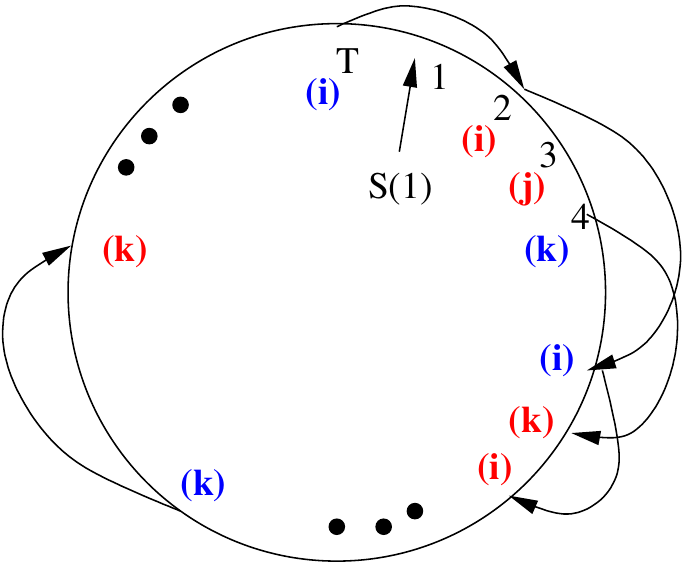}
\caption{Representing an improvement loop on a circle: times of updates $t$ and the updating user $(u(t))$ are illustrated along with their color right before a change. An arrow connects a single user's two consecutive color changes. We show that such improvement loop is not possible.}
\label{fig-circle}
\end{figure}

Since this sequence of updates is an improvement path, each change must increase the payoff of the user making the change\footnote{Here we assume that a user only makes a change if there is {\em strict} increase in its payoff.}. For example, suppose user $i$ changes from red to blue at time $t$, and $i$ has $x$ red neighbors and $y$ blue neighbors at $t$.\footnote{Since the users update their strategies in an asynchronous fashion, $x$ and $y$ do not change between $t^{-}$ and $t^{+}$.} Then we must have:
\begin{align}
g_b^{i}(y+1) > g_r^{i}(x+1)~.
\end{align}
Similarly, we can obtain one inequality for each of the $T$ changes. We shall show that these $T$ inequalities cannot be consistent with each other.
The challenge here is that this contradiction has to hold for arbitrary non-increasing functions $\{g_r^{i},g_b^{i}\}$.  The way we address this challenge is to show that the above inequality leads to another inequality that does {\em not} involve the payoff function when we consider pairs of reverse changes by the same user.  {The following definition will be useful for the proof.}


\begin{definition}[Reverse-change pairs]
Consider an arbitrary user $i$'s two reverse color changes in an improvement path, one from $s$ to $\bar{s}$ at time $t$ and the other from $\bar{s}$ to $s$ at time $t'$. Let $\mathcal{SS}_{t,t'}^i$ denote the set of $i$'s neighbors (not including $i$) who have the same color as $i$ at both times of change (i.e., at $t^{-}$ and $t'^{-}$, respectively).  Let $\mathcal{OO}_{t,t'}^i$ denote the set of $i$'s neighbors (not including $i$) who have the opposite color as $i$ at both times of change. Similarly, we will denote by $\mathcal{SO}_{t,t'}^i$ (respectively $\mathcal{OS}_{t,t'}^i$) the number of $i$'s neighbors whose color is the same as (opposite of, respectively) $i$'s at the first update and the opposite of (same as, respectively) $i$'s at the second update.
\label{def:SSOO}
\end{definition}


\begin{lemma}\label{lemma:reversechange}
 {\bf (Reverse-change inequality)}
Consider {a} spatial congestion game {with} two colors.  Suppose an arbitrary user $i$ makes two reverse color changes in an improvement path, one from $s$ to $\bar{s}$ at time $t$ and the other from $\bar{s}$ to $s$ at time $t'$.  Then we have
\begin{align}
|\mathcal{SS}_{t,t'}^i| > |\mathcal{OO}_{t,t'}^i|, ~~~ \forall i\in\cal{I} ~.
\end{align}
That is, among $i$'s neighbors, there are strictly more users \rev{with} the same color as $i$ at both times of change than those with the opposite color as $i$ at both times of change.
\end{lemma}
\begin{proof} Since this is an improvement path, whenever $i$ makes a change it is for higher payoff.  Thus we must have at the time of its first change and its second change, respectively, the following inequalities:
\begin{align}
g^i_{\bar{s}}(|\mathcal{OS}_{t,t'}^i|+|\mathcal{OO}_{t,t'}^i|+1) &> g^i_{s}(|\mathcal{SO}_{t,t'}^i|+|\mathcal{SS}_{t,t'}^i|+1) ~; \label{eqn-lemma1-1}\\
g^i_s(|\mathcal{SO}_{t,t'}^i|+|\mathcal{OO}_{t,t'}^i|+1) &> g^i_{\bar{s}}(|\mathcal{OS}_{t,t'}^i|+|\mathcal{SS}_{t,t'}^i|+1) ~. \label{eqn-lemma1-2}
\end{align}
We now prove the lemma by contradiction.  Suppose that the statement is not true and that we have $|\mathcal{SS}_{t,t'}^i| \le |\mathcal{OO}_{t,t'}^i|$.  Then due to the non-increasing assumption on the payoff functions we have
\begin{eqnarray*}
g^i_{\bar{s}}(|\mathcal{OS}_{t,t'}^i|+|\mathcal{SS}_{t,t'}^i|+1) &\geq&  g^i_{\bar{s}}(|\mathcal{OS}_{t,t'}^i|+|\mathcal{OO}_{t,t'}^i|+1) \\
&>& g^i_s(|\mathcal{SO}_{t,t'}^i|+|\mathcal{SS}_{t,t'}^i|+1) \\
&\geq& g^i_s(|\mathcal{SO}_{t,t'}^i|+|\mathcal{OO}_{t,t'}^i|+ 1)
\end{eqnarray*}
where 
the second inequality is due to (\ref{eqn-lemma1-1}).  This however contradicts with (\ref{eqn-lemma1-2}) and thus completes the proof.
\end{proof}

We point out that by Lemma \ref{lemma:reversechange} the payoff comparison is reduced to counting different sets of users.  This greatly simplifies the process of proving the main theorem of this section.  Below we show that it is impossible to have a finite sequence of asynchronous improvement steps ending in the same color state as it started with.  At the heart of the proof is the repeated use of Lemma \ref{lemma:reversechange} to show that loops cannot form in a sequence of asynchronous updates.

\begin{theorem}\label{thm:fip-2color}
{Every spatial congestion game with only two colors has the finite improvement property.}
\end{theorem}
\begin{proof}
We prove this by contradiction.  As illustrated by Figure~\ref{fig-circle}, we consider a sequence of improvement updates that results in the same state. 

Consider every two successive color changes, along this circle clockwise starting from time $t=1$, that a user $u(t)$ makes at time $t$ and $t'$ from color $s=s_{u(t)}(t)$ to $\bar{s}$, and then back to $s$, respectively.  Note that this will include the two ``successive'' changes formed by a user's last change and its first change (successive on this circle but not in terms of time).  We have illustrated this in Figure \ref{fig-circle} by connecting a pair of successive color changes using an arrow. It is easy to see that there are altogether $T$ such pairs (or arrows).

For each arrow in Figure \ref{fig-circle}, or equivalently each pair of successive color changes by the same user, we consider the two sets $\mathcal{SS}_{t,t'}^{u(t)}$ and $\mathcal{OO}_{t,t'}^{u(t)}$ in Definition~\ref{def:SSOO}. Due to the user association, we will also refer to these sets as {\em perceived} by user $u(t)$.  
By Lemma 1, given an updating sequence with the same starting and ending states, we have for each pair of successive reverse changes by the same user, at time $t$ and time $t'$, respectively:
\begin{eqnarray}\label{eqn-inequalities}
|\mathcal{SS}_{t,t'}^{u(t)}| > |\mathcal{OO}_{t,t'}^{u(t)}|, 
~~ t=1, 2, \cdots, T ~.
\end{eqnarray}
That is, the $\mathcal{SS}$ sets are strictly larger than the $\mathcal{OO}$ sets.  

This gives us a total of $T$ inequalities, one for each update in the sequence and each containing two sets.  Equivalently there is one inequality per arrow illustrated in Figure \ref{fig-circle}.
We next consider how many users are in each of these $2T$ sets (note that by keeping the same ``$>$'' relationship, the $\mathcal{SS}$ sets are always on the LHS of these inequalities and the $\mathcal{OO}$ sets are always on the RHS).  To do this, we will examine users by pairs -- we will take a pair of users and see how many times they appear in each other's sets in these inequalities.  \rev{We will use the following lemma.}
\begin{lemma}
Consider a pair of users $A$ and $B$ in an improvement updating loop, and consider how they are perceived in each other's set.  Then $A$ and $B$ collectively appear the same number of times in the LHS sets (the $\mathcal{SS}$ sets) and in the RHS sets (the $\mathcal{OO}$ sets).
\label{claim}
\end{lemma}

The proof of Lemma~\ref{claim} is given in Appendix~\ref{sec:claim}. 
\rev{Applying to all users,} Lemma~\ref{claim} implies that these users collectively contribute to an equal number of times to the LHS and RHS of the set of inequalities given in Eqn. (\ref{eqn-inequalities}). Adding up all these inequalities, this translates to the fact that the total size of the sets on the LHS and those on the RHS must be equal.  This however contradicts the strict inequality, thus completing the proof of Theorem \ref{thm:fip-2color}. \end{proof}

Theorem \ref{thm:fip-2color} establishes that when there are only two resources (colors), the FIP property holds, and consequently an NE exists.  
It turns out that this result does not in general hold when there are 3 or more resources/colors.  A counter-example is provided in the Appendix~\ref{appendix:3res} to illustrate this point.  This also implies that with 3 or more resources/colors,  an exact potential function does not exist for this game, as the FIP is a direct consequence of the existence of a potential function.

\subsection{The Finite Improvement Property for Identical Resources for Each User}

{The next theorem shows the second case in which the FIP property holds, \rev{when} all resources are identical to each user, but different users can have different payoff functions. \rev{This} can represent the case where all channels have the same bandwidth and same channel quality to each user (e.g, either with frequency flat fading or with proper channel interleaving such as IEEE 802.16d/e standard \cite{IEEE80216}), \rev{but users may have different channel conditions.}

\begin{theorem}\label{thm:fip-identical-resource}
For a spatial congestion game,
if for all $r\in {\cal R}$,  $i\in{\cal I}$, and $n\in \left\{1,\ldots,N\right\}$, we have $g_r^{i}(n)=g^{i}(n)$, 
then the game has the finite improvement property.
\end{theorem}

\begin{proof}
We prove this theorem by using a potential function argument.
Recall that user $i$'s total payoff under the strategy profile $\boldsymbol{\sigma}$ is given by 
$g^{i}(\boldsymbol{\sigma}) = g(n^{i}(\boldsymbol{\sigma})+1)$, 
with $n^{i}(\boldsymbol{\sigma}) = |\{j: \sigma_j = \sigma_i, j\in {\mathcal K}_i\}|$, 
where $\sigma_i \in {\cal R}$, and we have suppressed the subscript $r$ since all resources are identical. 

Now consider the following function defined on the strategy profile space:
\begin{eqnarray}\label{eqn:fip-identical-resource-potential}
\phi(\boldsymbol{\sigma}) = \sum_{i, j \in{\mathcal K}} \boldsymbol{1}(i\in {\mathcal K}_j) \boldsymbol{1}(\sigma_i=\sigma_j) = \frac{1}{2} \sum_{i\in {\mathcal K}} n^{i}(\boldsymbol{\sigma}) ~, 
\end{eqnarray}
where 
the indicator function $\boldsymbol{1}(A)=1$ if $A$ is true and $0$ otherwise.  For a particular strategy profile $\boldsymbol{\sigma}$, this function $\phi$ is the sum of all pairs of users that are connected (neighbors of each other) and have chosen the same resource under this strategy profile.  Viewed in a graph, this function is the total number of edges connecting nodes with the same color.

We see that every time user $i$ improves its payoff by switching from strategy $\sigma_i$ to $\sigma_i^{'}$ and thus reducing $n^{i}(\boldsymbol{\sigma}^{-i}, \sigma_i)$ to $n^{i}(\boldsymbol{\sigma}^{-i}, \sigma_i^{'})$ (as $g^{i}$ is a non-increasing function), the value of $\phi(\cdot)$ strictly decreases accordingly \footnote{It's easy to see that a non-increasing function $G(\sum_{i, j \in{\mathcal K}} \boldsymbol{1}(i\in {\mathcal K}_j) \boldsymbol{1}(\sigma_i=\sigma_j))$ is an ordinal potential function of this game, as its value improves each time a user's individual payoff is improved (which decreases the value of its argument).}.
%
%
As this function is bounded from below, 
the game has the FIP property, and this process eventually converges to a fixed point which is a Nash Equilibrium.
 \end{proof}


\subsection{The Finite Improvement Property for Complete Graphs}

We end this section by stating that an SCG defined over a fully connected graph always has the FIP property: \rev{SCG over a complete graph simply reduces to the standard CG, thus the result}. 
\begin{theorem}
{When the graph is complete, the associated SCG has the FIP property and thus a NE always exists.}
\end{theorem}



\section{Existence of a Pure Strategy Nash Equilibrium}
\label{sec-graph}

{The FIP property guarantees the existence of NE, but such property may not exist in general. In this section, we examine what graph properties or user payoff functions will guarantee the existence of a pure strategy NE in the absence of the FIP property.}


Specifically, we show that a pure NE always exists for SCGs defined on graphs that {are in the form of a tree or in the form of a loop. We also show the existence of a pure NE when the graph is regular, bipartite, and \rev{payoff functions are non-user specific}.  
We also give counter examples in Appendices~\ref{appendix:nonmonto} and \ref{appendix:direct} that a pure strategy NE does not generally exist when the payoff functions are non-monotonic or when the network graph is directed. 
%



{\subsection{Existence of NE on a Tree Graph}}

%

We show that a pure strategy NE exists when the underlying network graph is given by a tree.
We denote by $G_N$ the underlying network (graph) of the $N$-player SCG $\Gamma_N$.
The payoff functions $g_r^i(n_r^i)$ are non-increasing, and $n_r^{i}(\boldsymbol{\sigma})$ denotes the number of neighbors of user/player $i$ (excluding $i$) using strategy $r$. 

\begin{lemma}\label{lem:one-link}
If every $N$-player SCG  $\Gamma_N$ has at least one pure strategy NE, then every $(N+1)$-player SCG  $\Gamma_{N+1}$ formed by connecting a new player to \rev{an existing} player in a $N$-player network $G_N$ has at least one pure strategy NE.
\end{lemma}

\begin{proof}
By assumption $\Gamma_N$ has a pure strategy NE denoted by $\boldsymbol{\sigma}=\{\sigma_1, \sigma_2, \cdots, \sigma_N\}$.  Suppose $\Gamma_N$ is in such an NE.  Now connect new player $N+1$ to an arbitrary player $j$ in $G_N$.
This is illustrated in Figure \ref{fig:single-link}.
\begin{figure}[h]
\centering
\includegraphics[scale=0.7]{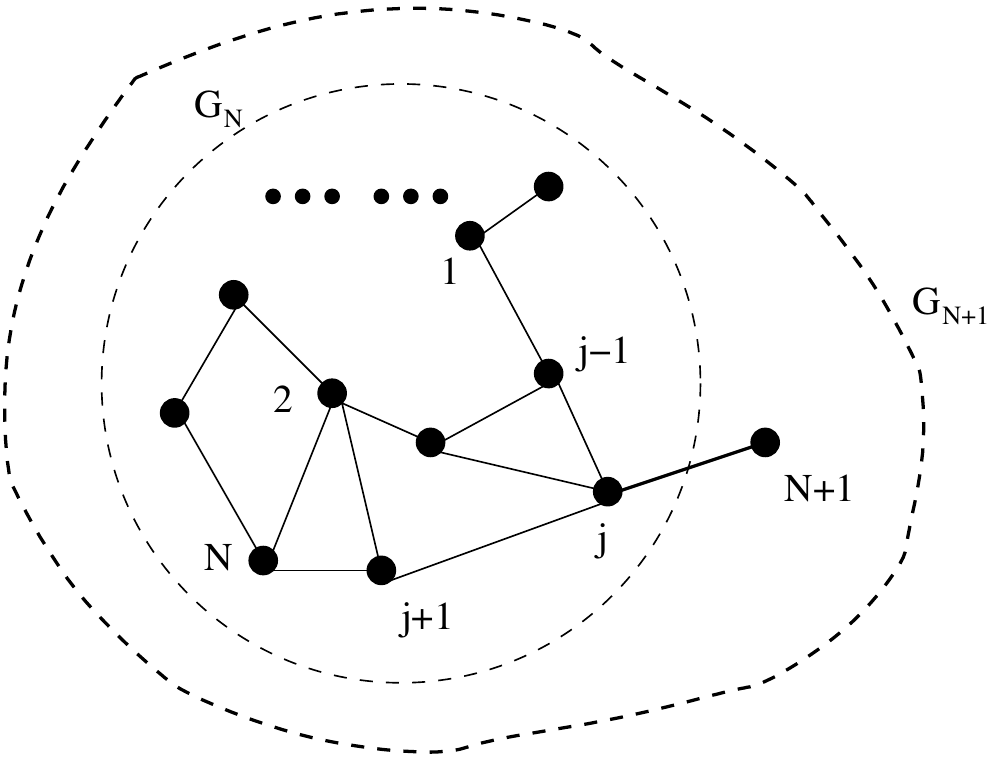}
\caption{Adding one more player to the network $G_N$ with a single link.}
\label{fig:single-link}
\end{figure}

Let player $N+1$ select its best response strategy:
\begin{equation*}
\sigma_{N+1} = r_o = \text{argmax}_{r\in \cal{R}} g_r^{N+1}\left(n_r^{N+1}(\boldsymbol{\sigma})+1\right),
\end{equation*}
where $n_r^{N+1}$ is defined on the extended network $G_{N+1}$, and takes on the value of 1 or 0 depending on whether player $j$ selects strategy $r$ or not.
We now consider three cases depending on $j$'s strategy change in response to the {network expansion} from $G_N$ to $G_{N+1}$. 


Case 1: $\sigma_j \neq r_o$.  In this case, player  $N+1$ selected a resource different from $j$'s, so $j$ has no incentive to change its strategy in response to the addition of player $N+1$.  In turn player $N+1$ will remain in $r_o$ as this is its best response, and no other players are affected by this single-link network extension.  Thus the strategy profile $(\sigma_1, \cdots, \sigma_N, r_o)$ is a pure strategy NE for the game $\Gamma_{N+1}$.

Case 2: $\sigma_j = \sigma_{N+1} = r_o$, and player $j$'s best response to the {network expansion}
remains $\sigma_j=r_o$.  That is, even with the additional interfering neighbor $N+1$, the best choice for $j$ remains $r_o$.  In this case again we reach a pure strategy NE for the game $\Gamma_{N+1}$ with the same argument as in Case 1.

Case 3: $\sigma_{j} = \sigma_{N+1} = r_o$, and player $j$'s best response to this {network expansion} is to move away from strategy $r_o$.  In this case more players may in turn change strategies.  Suppose we hold player $(N+1)$'s strategy fixed at $r_o$.  Consider now a new \rev{$N$-player SCG} $\bar\Gamma_N$, defined on the original network $G_{N}$, but with the following modified payoff functions for $r\in{\cal R}$ and $i\in {\cal I}$:
\begin{eqnarray*}
\bar{g}_{r}^{i}(n_r^i + 1) = \left\{\begin{array}{ll}
g_r^{i}(n_r^i+2) & \text{if } i=j, r=r_o \\
g_r^{i}(n_r^i+1) & \text{otherwise}
\end{array}\right. ~.
\end{eqnarray*}
In words, the game $\bar\Gamma_N$ is almost the same as the original game $\Gamma_{N}$, the only difference being that the addition of player $(N+1)$ and its strategy $r_o$ is built into player $j$'s modified payoff function.
%
By assumption of Lemma~\ref{lem:one-link}, this game with $N$ players has a pure strategy NE and we denote that by $\boldsymbol{\bar{\sigma}}$.  Suppose $\boldsymbol{\bar{\sigma}}$ is reached in the network $G_N$ with player $(N+1)$ fixed at $\sigma_{N+1}=r_o$. If we have $\bar\sigma_j = r_o$, then obviously player $(N+1)$ has no incentive to change its strategy because as far as it is concerned its environment has not changed.  In turn no player in $G_N$ will change its strategy because they are already in an NE with player $(N+1)$ held at $r_o$.  If $\bar\sigma_j \neq r_o$, then player $(N+1)$ has even less incentive to change its strategy because 
$j$ moved away from $r_o$ which does not decrease player $N+1$'s payoff on this resource, and at the same time its payoff for using any other resource is no better.  Again $r_o$ is player $(N+1)$'s best response.
In either case,  strategy profile $(\boldsymbol{\bar{\sigma}}, r_o)$ is a new NE for the game $\Gamma_{N+1}$.
\end{proof}

\begin{remark}
Note that in the above lemma, the network $G_N$ itself does not have to be a tree.  The lemma states that as long as an NE exists for {one class of networks}, then by adding one more node through a single link, an NE exists in the new network.
\end{remark}

\begin{theorem}\label{thm:tree}
Any SCG defined over a {\em tree} has at least one pure strategy NE.
\end{theorem}

\begin{proof}
The proof is easily obtained by noting that any tree can be constructed by starting from a single node and adding one node (connected through a single link) at a time.  Formally, we prove this by induction.   Start with a single player indexed by $1$. This game has a pure strategy NE, in which the player selects $\sigma_1=\text{argmax}_{r\in \cal{R}} g_r^1(1)$ for any payoff functions.  Assume that any $N$-player game $\Gamma_N$ over a tree $G_N$ with any set of non-increasing payoff functions has at least one pure strategy NE.   Any tree $G_{N+1}$ may be constructed by adding one more leaf node to some other tree $G_N$ by connecting it to only one of the players in $G_N$.   Lemma \ref{lem:one-link} guarantees that such a formation will result in a game with at least one pure strategy NE.
\end{proof}

{\subsection{Existence of NE on a Loop}}

\begin{theorem}\label{thm:loop}
Any SCG defined over a {\em loop} network has at least one pure strategy NE.
\end{theorem}

\begin{proof}[Proof (Sketch)]
\rev{The detailed and complete proof of this theorem can be found in Appendix \ref{appendix:loop}.} We begin this proof by assuming that every player \rev{on} the loop always has a unique best response. This will always be the case, unless equalities of the form $g_r ^i (x) = g_{r'} ^i (x')$ cause two resources to be tied as $i$'s best response. {Even in the tie case, we can still get a unique best response by assuming that each user has a preference order among colors when the payoffs are the same.\footnote{For example, a user with a color preference of ``red$>$blue$>$green'' will pick red if the payoffs of choosing red or blue are the same.} In fact, our assumption does not \rev{affect} the validity of the proof, because relaxing it only widens the set of NE a given game on the loop has.}

Under our assumption, we show that every player $i$ can be associated with a triple $(a(i),b(i),c(i)) \in \mathcal{R} ^3$ of possible best responses to different scenarios. The triple has the following properties.

\begin{enumerate}

\item If $i$ has no neighbors playing $a(i)$, then $i$'s best response is $a(i)$, where $a(i) = \argmax _{i \in \mathcal{R}} (g_r ^i (1) )$.

\item If $i$ has one neighbor playing $a(i)$, with the other neighbor not playing $b(i)$, then $i$'s best response is to play $b(i)$.

\item If $i$ has one neighbor playing $a(i)$ and one neighbor playing $b(i)$, then $i$'s best response is $c(i)$.
\end{enumerate}

The main idea of the proof is to show the existence of NE given the existence of players with various kinds of triples.
We start by showing that if there exists a player $i^*$ such that $a(i^*) = b(i^*)$, {then an NE exists.} The way to show this is to hold $i^*$ fixed playing $a(i^*)$ and let the other players alter their strategies freely. Since the other players are essentially playing on a line graph (which is a type of tree graph) we use theorem \ref{thm:tree} to construct a strategy configuration within which each player in $\mathcal{I} - \{ i \}$ employs their best response. We then show that allowing $i^*$ to employ its best response under this configuration constitutes an NE.

Next we show that if no such player $i^*$ exists (so that $a(i) \neq b(i), \forall i$), an NE must also exist. This is done by constructing an algorithm which produces strategy configurations that satisfy many of the players around the loop.
The algorithm begins by assigning player $1$ a strategy $\sigma _1 \in \{ a(1) , b(1) \}$. After this, the algorithm continues to allocate strategies $\sigma _i$ to $i \in \{ 2 , 3 ,..., N \}$ in such a way that $\sigma_i = a(i)$ unless $a(i) = \sigma _{i-1}$ in which case $\sigma _i = b(i)$.
We use this algorithm repeatedly to demonstrate the existence of NE under several cases. The entire set of cases we consider exhausts all the possible games where $a(i) \neq b(i), \forall i$.
%
%
\end{proof}


{\subsection{Existence of NE on a Regular Bipartite Graph}}

A graph is regular when all its vertices have the same number of connections. A graph is bipartite when its vertices can be colored red and blue {(only two colors)} so that no edge connects a pair of vertices with the same color.
Many well known graphs are regular and bipartite including hypercubes and rectangular lattices.

\begin{figure}[h]
\centering
\includegraphics[scale=0.4]{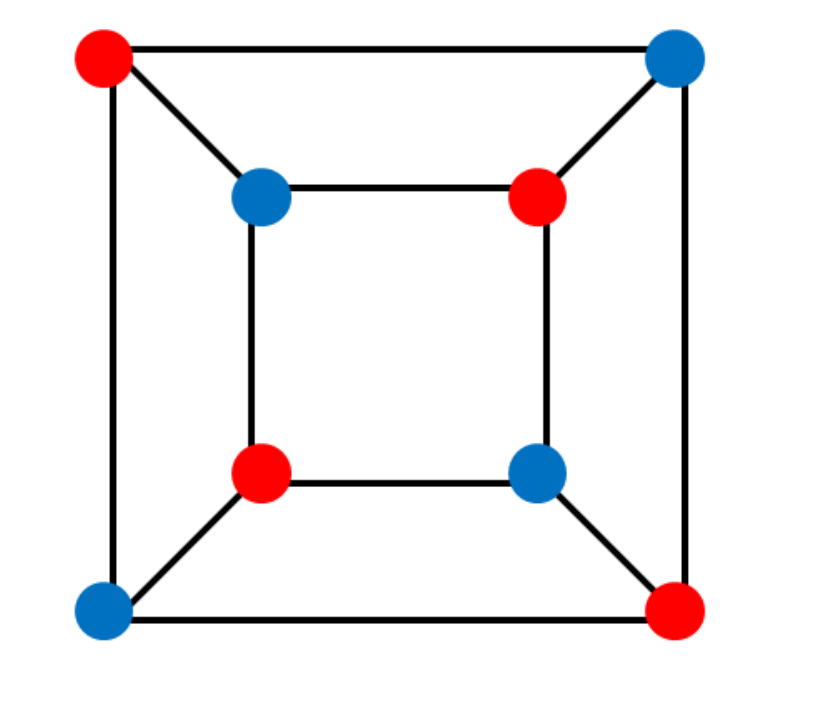}
\caption{The cube graph is regular and bipartite.}
\label{fig:cube}
\end{figure}

\begin{theorem}\label{thm:bipartite}

If the network is regular and bipartite and payoff functions are \rev{non-user specific}, 
then there always exists a pure strategy Nash equilibrium.

\end{theorem}

\begin{proof}
As payoff functions are not user-specific, we will suppress the superscript $i$ in the function $g_r^{i}(\cdot)$.
Suppose the graph is bipartite and each vertex has {degree $d$ (so $d$ denotes the number of connections each vertex has, e.g., $d=3$  in Fig.~\ref{fig:cube})}.  {Without loss of generality, we order the resources such that the payoff functions satisfy  $g_1 (1) \geq g_2 (1) \geq ... \geq g_R (1)$}. \rev{If $g_r (d+1) \geq g_b (1)$}, then resource $r$ dominates and we can trivially construct an NE by allowing each player to use resource $r$.

Now consider the case where \rev{$g_r (d+1) < g_b (1)$}. Since our graph is bipartite, we may color the vertices with ``colors'' \rev{$r$ and $b$} in such a way that no edge connects a pair of vertices with the same color. We can think of this coloring as a resource allocation $\boldsymbol\sigma$. Under this allocation each employer of \rev{$b$} will receive payoff \rev{$g_b (1)$} (because they have no neighbors employing \rev{$b$}) whereas they \emph{would} get \rev{$g_r (d+1) \leq g_b (1)$} if they played \rev{$r$}, which is no better. So each employer of \rev{$b$} is playing its best response under $\boldsymbol\sigma$. In a similar way, the fact that \rev{$g_r (1) \geq g_b (1) \geq g_r (d+1)$} implies that each employer of \rev{$r$} is playing its best response.
\end{proof}

{We end this section by noting that regardless of the type of graphs, whenever there is a dominant resource $r$, i.e., its payoff function is such that $g_r^i(K_{d}+1) \geq g_{r^{\prime}}^i(1)$, where $K_d=\max\{|\mathcal{K}_i|, i=1, 2, \cdots, N\}$, for all $r^{\prime} \in \mathcal{R}$ and all $i\in{\cal I}$, then a NE obviously exists where all users share the same dominant resource.
}

%
%


\section{Discussion}\label{sec-discussion}
{
While the results derived in this paper present original contributions to the body of knowledge on congestion games, the spatial congestion game has its advantages and limitations as a model in the context of {wireless multi-cahnnel networks}.
%
In this section we discuss in more details the relevance of the results obtained here as well as possible directions of future studies.



Two results obtained in this paper are of particular interest, namely Theorem \ref{thm:fip-2color} and Theorem \ref{thm:fip-identical-resource}.  Theorem \ref{thm:fip-2color} showed that when users are limited to only two channels, the finite improvement property holds over arbitrary graphs with user-specific payoff functions.  Theorem \ref{thm:fip-identical-resource} showed that when channels are of equal width and propagation characteristics for each user (as is the case when a contiguous block of bandwidth is evenly sliced into smaller channels), the finite improvement property holds. This is true  even if the channels are of different quality to different users, e.g., due to the use of different modulation schemes.  This latter scenario is a very realistic one, as this is the case with multiple channels in WiFi (IEEE 802.11b), bluetooth, and so on.  The finite improvement property suggests that in such systems greedy user updates will lead to an NE, which is the local minimizer of the explicit potential function (Eqn (\ref{eqn:fip-identical-resource-potential}) in this case).  
This means that we not only have an easy way of obtaining an NE, but also have a sense of the (local) efficiency of this NE.

To precisely assess the optimality of an NE, a commonly adopted approach is to \rev{characterize} what's known as the price of anarchy (PoA). The PoA characterizes the ``distance'' between the NE and the social optimal solution of the system. One of the early \rev{results} along this line was \cite{EC1999}. The bounds of
PoA were proven in \cite{CK2005} and \cite{AD2006} for both linear
and polynomial cost functions. {Recent work such as \cite{T2009}
gave the exact PoA for a class of congestion games. It identified a
sufficient condition for an upper-bound} and later showed that the
bound is achievable. Reference \cite{LawHuang2009} computed the exact PoAs for congestion games with player-specific payoff functions in the context of cognitive radio spectrum sharing. None of the existing PoA literature studied the spatial congestion game as we proposed in this paper.



One limitation of our spatial congestion game model is that it treats all interference relationships equally, i.e., the underlying network graph is unweighted.  In reality the channel quality perceived by a user depends not only on whoelse is using the same channel and can potentially interfere, but also its distances to these interfering users.  One way to address this is to define the congestion game over a weighted network graph, and define the user payoff as a function of the weights on links connecting interfering users who use the same channel. Analysis along this line will be very interesting yet challenging.

Throughout our discussion, we have limited our attention to the case where each user can access one resource/channel at a time.  In reality it's also possible for a user to access multiple channels at a time.
As mentioned earlier, if all users can access all channels simultaneously and the available transmission power is decoupled across the channels, then the resulting congestion game is not particularly interesting, as an obvious NE is where all users use all the resources.
A more interesting case is when users are limited to the number of channels they can access simultaneously.  An additional feature may be that different users have different sets of channels they are allowed to access, i.e., user $i$'s strategy space $\sigma_i \subset 2^{{\cal R}_i}$, where ${\cal R}_i \subset {\cal R}$ is user $i$'s set of allowed channels.  Finally, a user may need to spread the communication resource such as transmission power among multiple channels, thus transmitting over one or multiple channels implies different payoff functions for each channel. All these features will make the resulting game much more complicated and are subjects of future study.
}

\section{Conclusion}
\label{sec-conclusion}

In this paper, we considered an extension to the classical congestion games by allowing resources to be reused among non-interfering users.   {This is a more appropriate model to use in the context of spectrum sharing in multi-channel wireless networks, where spatial reuse is frequently exploited to increase spectrum utilization due to decay of wireless signals distance.}

{The resulting game, spatial congestion game, is a generalization to the original congestion game.  We have shown that the finite improvement property (FIP) holds when there are only two resources or the resources are identical to each user (but may be different between users). The FIP guarantees the existence of a pure strategy NE. 
We also show that a pure strategy NE exists without the FIP if the network can be modeled by a tree graph, a} {regular} {bipartite graph, a loop, or with a dominate resource.}




\appendix

\subsection{Proof of Lemma~\ref{claim}} \label{sec:claim}
%
First note that $A$ and $B$ have to be in each other's interference set for them to appear in each other's $\mathcal{SS}$ and $\mathcal{OO}$ sets.  Since we are only looking at two users and how they appear in each other's sets, without loss of generality we can limit our attention to a subsequence of the original updating sequence involving only $A$ and $B$, given by
\begin{eqnarray}
U_{AB} = \{ u(t_1), u(t_2), \cdots, u(t_l)\}
\end{eqnarray}
where $u(t_j) \in \{A, B\}$, $t_j \in\{1, 2, \cdots, T\}$, and $l$ is the length of this subsequence, i.e., the total number of updates between $A$ and $B$.  As before, this subsequence can also be represented clockwise along a circle.

It helps to consider an example of such a sequence, say, $ABAABBABAA$, also shown in Figure \ref{fig-example}.  In what follows we will \rev{refer to an ``odd train''} as the odd number of consecutive changes of one user sandwiched between the other user's changes, e.g., the odd train $ABA$ in the above subsequence. To avoid ambiguity, we will further write this sequence as $A_1 B_2 A_3 A_4 B_5 B_6 A_7 B_8 A_9 A_{10}$.
\begin{figure}[h]
\centering
\includegraphics[width=2.3in]{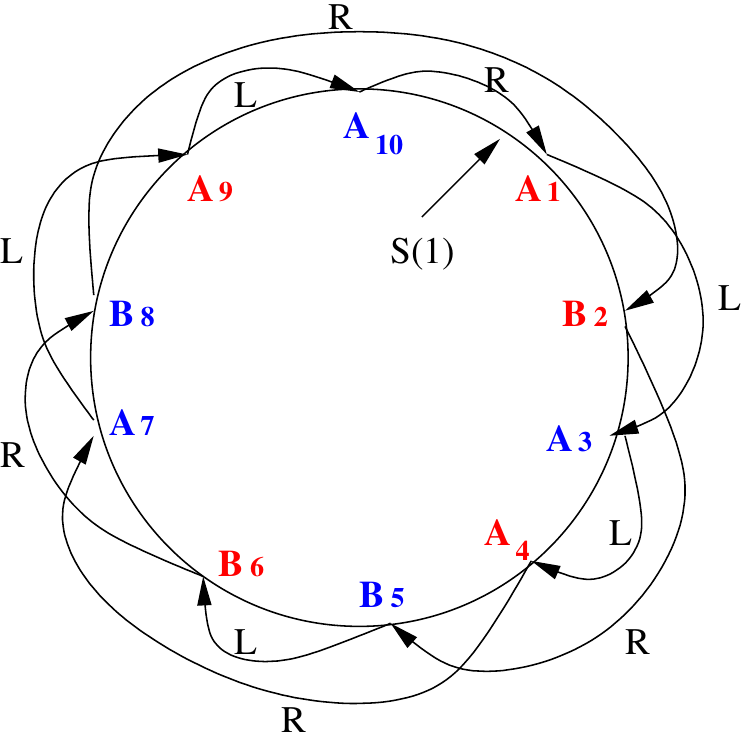}
\caption{Example of an updating sequence ``$A_1 B_2 A_3 A_4 B_5 B_6 A_7 B_8 A_9 A_{10}$'' illustrated on a circle. The color coding denotes the color of a user right before the indicated change.  Each arrow connecting two successive changes by the same user induces an inequality perceived by this user.  The labels ``L'' and ``R'' on an arrow indicate to which side of this inequality (LHS and RHS respectively) the other user contributes to.  As can be seen the labels alternates in each subsequent inequality.}
\label{fig-example}
\end{figure}

A few things to note about such a sequence:
\begin{enumerate}
\item Since the starting and ending states are the same, each user must appear an even number of times in the sequence.
\rev{Consequently} there must be an even number of odd trains along the circle for any user. 

\item A user (say $A$) only appears in the other's (say $B$'s) $\mathcal{SS}$ or $\mathcal{OO}$ sets if it has an odd train between the other user's two successive appearances.   This means that there is an even number of relevant inequalities where $A$ appears in $B$'s inequalities (either on the LHS or the RHS), and vice versa.

\item Consider the collection of all relevant inequalities discussed above, one for each odd train, in the order of their appearance on the circle (all four such inequalities are illustrated in Figure \ref{fig-example}).  Then $A$ and $B$ contribute to each other's inequalities on alternating sides along this updating sequence/circle.  That is, suppose the first inequality is $A$'s and $B$ goes into its LHS, then in the next inequality (could be either $A$'s or $B$'s) the contribution (either $A$ to $B$'s inequality or $B$ to $A$'s inequality) is on the RHS.  Take our running example, for instance, the first inequality is due to the odd train marked by the sequence $A_1 B_2 A_3$, and the second $B_6 A_7 B_8$.  Suppose $A$ and $B$ start with different colors, then in the first inequality, $B$ appears in the RHS; in the second, $A$ appears in the LHS.
\end{enumerate}

We now explain why the third point above is true.
The reason is because for one user ($B$) to appear in the other's ($A$'s) LHS, they must start by having the same color and again have the same color right before $A$'s second change (see, e.g., the subsequence $A_1 B_2 A_3$ in the running example).  Until the next odd train ($B_6 A_7 B_8$), both will make an even number of changes including $A$'s second change ($A_3 A_4 B_5 B_6$).  The next inequality belongs to the user who makes the last change before the next odd train ($B$).  As perceived by this user ($B$) right before this change, the two must now have different colors.  This is because as just stated $A$ will have made an even number of changes from the last time they are of the same color (by the end of $A_1 B_2$), while $B$ is exactly one change away from an even number of changes (by the end of $A_1 B_2 A_3 A_4 B_5$).  Therefore, the contribution from the other user ($A$) to this inequality must be to the RHS. 

To summarize, one can see that essentially the color relationship between $A$ and $B$ reverses upon each update, and there is an odd number of updates between the starting points of two consecutive odd trains (e.g., 5 updates between $A_1$ and $B_6$, or 1 update between $B_6$ and $A_7$) so the color relationship flips for each inequality in sequence.

The above argument establishes that as we go down the list of inequalities and count the size of the sets on the LHS vs. that on the RHS, we alternate between the two sides.  Since there are exactly even number of such inequalities, we have established that $A$ and $B$ collectively appear the same number of times in the LHS sets and in the RHS sets.\hfill \qed

\subsection{Counter-Example for 3 Resources}\label{appendix:3res}
The example below shows that the FIP property does not necessarily hold for when there are $3$ resources/colors.

\begin{example}
Suppose we have three colors to assign, denoted by $r$  (red), $p$ (purple), and $b$ (blue).  Consider a network topology shown in Figure \ref{fig:counter}, where we will primarily focus on nodes $A$, $B$, $C$ and $D$.
In addition to node $C$, node $A$ is also connected to $A_r$, $A_p$ and $A_b$ nodes of colors red, green and blue, respectively.   $B_r$, $B_p$, $B_b$, $C_r$, $C_p$, $C_b$, and $D_r$, $D_p$, $D_b$ and similarly defined and illustrated in Figure \ref{fig:counter}.
Note that these sets may not be disjoint, e.g., a single node may contribute to both $A_r$ and $B_r$, and so on.
\begin{figure}[h]
\centering%
\includegraphics[width=2.3in]{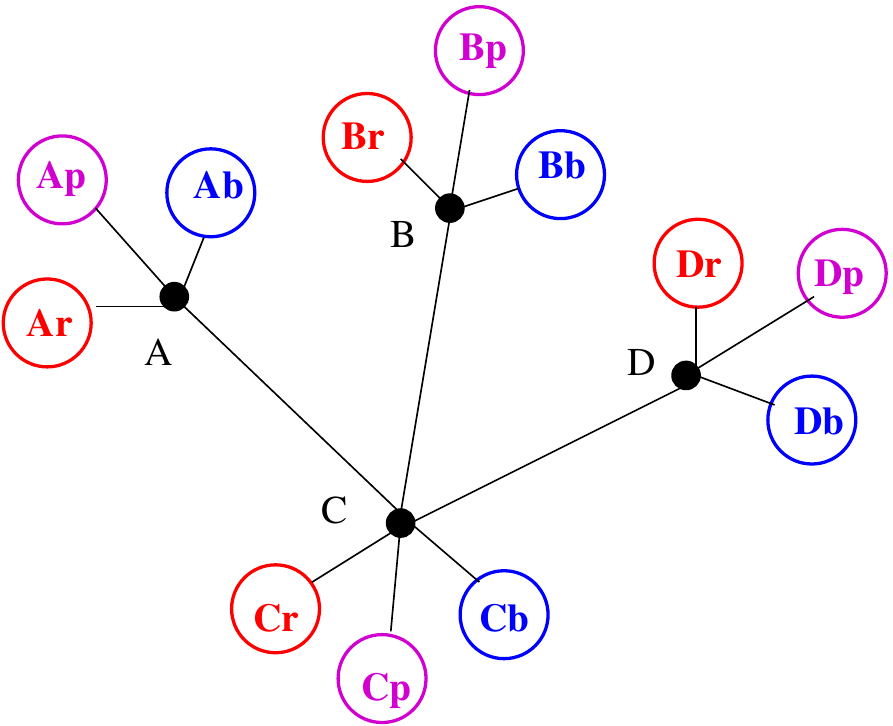}
\caption{A counter example of 3 colors. 
}
\label{fig:counter}
\end{figure}

Consider now the sequence of improvement updates shown in Table \ref{tbl:counter} involving only nodes $A$, $B$, $C$, and $D$, i.e., within this sequence none of the other nodes change color (note that this is possible in an asynchronous improvement path), where the notation $s_1\rightarrow s_2$ denotes a color change from $s_1$ to $s_2$. At time $0$, the initial color assignment is given. 

\begin{table}[h]\label{tbl:counter}
\begin{center}
\begin{tabular}{|r|c|c|c|c|}
\hline
time step & $A$ & $B$ & $C$ & $D$ \\
\hline
0 & b 			 & p & p & b \\
\hline
1 & b $\rightarrow$ r &    &     &  \\
\hline
2 &                            & p $\rightarrow$ r & & \\
\hline
3 &                            &                            &                   & b $\rightarrow$ r \\
\hline
4 & 				& 				& p $\rightarrow$ r &                 \\
\hline
5 & r $\rightarrow$ p &  & & \\
\hline
6 & 				& 				& 				& r $\rightarrow$ b \\
\hline
7 & 				& r $\rightarrow$ b & & \\
\hline
8 & 				& 				& r $\rightarrow$ b & \\
\hline
9 & p $\rightarrow$ b & & & \\
\hline
10& 				& 				& b $\rightarrow$ p & \\
\hline
11& 				& b $\rightarrow$ p & & \\
\hline
\end{tabular}\medskip
\caption{3-color counter example.}
\end{center}
\end{table}

\medskip
We see that this sequence of color changes form a loop, i.e., all nodes return to the same color they had when the loop started. If we can show that such loop is feasible, then we have found an counter example.
%
%
%
For this to be an improvement loop such that each color change results in an improved payoff, it suffices for the following sets of conditions to hold. Here we assume all users have the same payoff function and have suppressed the superscript $i$ in $g_r^{i}(\cdot)$, and the notation ``$>_{k}$'' denotes that the improvement occurs at time $k$.
\begin{eqnarray*}
&& g_r(A_r+1) >_1 g_b(A_b+1) > g_b(A_b+2) \\
&& >_{9}  g_p(A_p+1) >_{5} g_r(A_r+2) ~; \\
&& g_r(B_r+1) >_{2} g_p(B_p+2) >_{11} g_b(B_b+1)\\
&&>_{7} g_r(B_r+2) ~; \\
&& g_b(C_b+3) >_{8} g_r(C_r+1) > g_r(C_r+4) \\
&& >_{4} g_p(C_p+1) >_{10} g_b(C_b+4) ~; \\
&& g_r(D_r+1) >_{3} g_b(D_b+1) >_{6} g_r(D_r+2)
\end{eqnarray*}
It is straightforward to verify the sufficiency of these conditions by following a node's sequence of changes.

To complete this counter example, it remains to show that the above set of inequalities are feasible given appropriate choices of $A_x$, $B_x$, $C_x$ and $D_x$, $x\in\{r, p, b\}$.  There are many such choices; one example is $A_x = 5,  B_x = 3,  C_x = 7,  D_x = 1$, for all $x \in\{ r, p, b\}$.
With such a choice, and substituting them into the earlier set of inequalities and through proper reordering, we obtain the following single chain of inequalities:
\begin{eqnarray*}
&& g_r(2) > g_b(2) > g_r(3)
> g_r(4) > g_p(5) > g_b(4) \\
&>& g_r(5) > g_r(6) 
>  g_b(6) > g_b(7) > g_p(6) > g_r(7) \\
&>& g_b(10) > g_r(8) > g_r(11) > g_p(8)>g_b(11)
\end{eqnarray*}
It should be obvious that this chain of inequalities can be easily satisfied by the right choices of non-increasing payoff functions.
%
\end{example}

It is easy to see how if we have more than 3 colors, this loop will still be an improving loop as long as the above inequalities hold. This means that for 3 colors or more the FIP property does not hold in general.  Note that the updates in this example are not always best response updates;  {they can be better responses which still result in payoff improvements.}

\subsection{Counter Example of Non-monotonic Payoff Functions}\label{appendix:nonmonto}

Below we show that a pure strategy NE may not exist when the network graph is undirected but the payoff function is non-monotonic, even when they are non-user specific.

\begin{example}
Consider a 3-user, 2-resource network given in Figure \ref{fig:PNE4}. The payoff functions have the following property
\begin{eqnarray*}
g_2(2)>g_1(2)>g_2(1)> g_1(3)> g_1(1)> g_2(3) ~.
\end{eqnarray*}
One example of this is when $g_1(1)=2, g_1(2)=5, g_1(3)=3, g_2(1)=4, g_2(2)=6$, and $ g_2(3)=1$. The game matrix corresponding to these payoff functions are given below.  It is easy to verify that there exists no pure strategy NE.
\begin{figure}[h]
\centering
\includegraphics[width=1.2in]{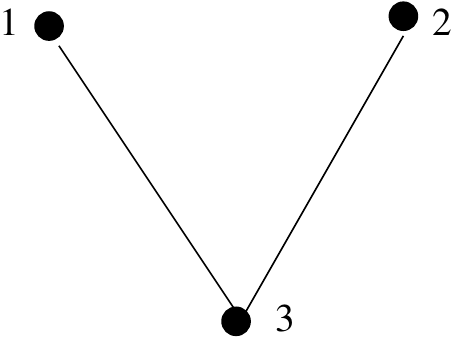}
\caption{Counter example of non-monotonic payoff functions}
\label{fig:PNE4}
\end{figure}

\medskip
\begin{center}
\begin{tabular}{|r|c|c|c|c|}
\hline
User 3 \slash User 1,2 & $(1,1)$ & $(1,2)$ & $(2,1)$ & $(2,2)$ \\
\hline
1 & 5, 5, 3 & 5, 4, 5 & 4, 5, 5 & 4, 4, 2 \\
\hline
2 & 2, 2, 4 & 2, 6, 6 & 6, 2, 6 & 6, 6, 1  \\
\hline
\end{tabular}
\end{center}
\end{example}

\subsection{Counter Example of a Directed Graph}\label{appendix:direct}
Below we show that a pure strategy NE may not exist when the network graph is directed.

\begin{example}
Consider a 4-user, 3-resource network given in Figure \ref{fig:PNE3}. It can be shown that a pure strategy NE does not exist when the payoff functions are non-increasing and have the following property.
\begin{eqnarray*}
&& g_3(1)> g_2(1)> g_2(2)> g_3(2)> g_1(1)  > g_1(2) \\
&>& g_2(3)> g_1(3)> g_2(4)> g_1(4)> g_3(3)> g_3(4) ~.
\end{eqnarray*}
\begin{figure}[h]
\centering
\includegraphics[width=1.4in]{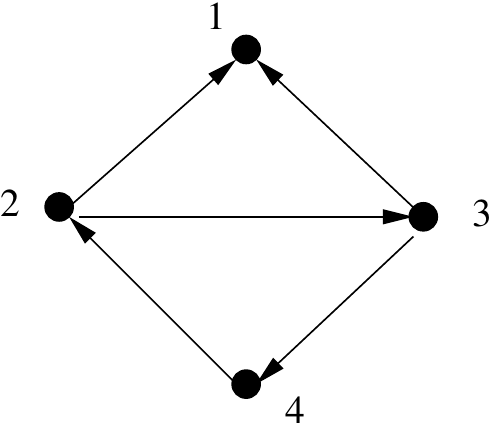}
\caption{Counter-example for directed graphs}
\label{fig:PNE3}
\end{figure}

We do not include an example game matrix for brevity.  We invite an interested reader to verify this example.
\end{example}

\subsection{Proof of Theorem \ref{thm:loop}}\label{appendix:loop}
%

Here we give a constructive proof that  {every \rev{SCG} on a loop graph has a NE}.

Suppose our set of players $\mathcal{I} = \{ 0 , 1 ,..., N-1 \}$ are labeled from $0$ to $N-1$. The interference graph forms a loop so that each $i \in \mathcal{I}$ has interference set $\mathcal{K} _i = \{ i-1 , i+1 \}$, where addition is performed modulo $N$ (throughout this proof). Suppose we have resource set $\mathcal{R} = \{ 1 , 2 ,..., R \}$ and each player $i$ is associated with non-increasing payoff functions $g_1 ^i , g_2 ^i, ..., g_R ^i$ overall resources $1$ to $R$.

A strategy allocation $\boldsymbol\sigma = ( \sigma _0, \sigma _1,...,\sigma _{N-1} )$ is an assignment of one resource (strategy) $\sigma _i \in \mathcal{R}$ to each $i \in \mathcal{I}$.
A player $i$ will select the strategy $i$ that maximize its payoff $g_r ^i ( 1 + |\{ j \in \mathcal{K}_i : \sigma_j = r \}|)$.
For a multi-set $S$ of elements from $\mathcal{R}$, let us define $\beta ^i (S) = \argmax_{r \in \mathcal{R}}  g_r ^i \left( 1 + |S \cap \{r \}|\right) $ to be the set of best responses that player $i$ has when $S$ is the multi-set of strategies allocated to its neighbors.

In this proof, we are concerned with games on loop graphs where each player has two neighbors. This means we are concerned with the values of $\beta ^i ( \{ a , b \} )$ with $a,b \in \mathcal{R}$.\footnote{Note that $\{ a , a \} \neq \{ a \}$ because we are discussing multi-sets.} Recall that a strategy allocation $\boldsymbol\sigma$ is a Nash equilibrium (NE) if and only if $\sigma _i = \beta ^i ( \{ \sigma _{i-1} , \sigma _{i+1} \} )$, $\forall i \in \mathcal{I}$.

Often $\beta ^i (S)$ will be a singleton (a single well defined best response). This will happen except when equalities of the form $g_r ^i (x) = g_{r'} ^i (x')$ cause multiple strategies/resources to be tied as best responses. In this proof we shall assume that such equalities do not occur so we can think of $\beta ^i$ as a map from $\mathcal{R} ^2$ to $\mathcal{R}$. We will show the existence of Nash equilibria under this restriction. Relaxing the restriction will only increase the number of options players have as their best responses and will hence maintain previously found Nash equilibria.

Now we shall examine the possible \emph{types} of players (i.e., the different forms that best response function $\beta ^i$ can take for different players).

For $i \in \mathcal{I}$ let $(a(i) , b(i) , c(i) )$ denote the $\emph{type}$ of player $i$ where: 
\begin{enumerate}
\item $a(i) = \beta ^i ( \emptyset ) = \argmax _{r \in \mathcal{R}} (g_r ^i (1))$,

\item  $b(i) = \beta ^i ( \{ a(i) \} )$

\item   $c(i) = \beta ^i ( \{ a(i) , b(i) \} )$.
\end{enumerate}

The way that a type $(a(i) , b(i) , c(i) )$ player's best responses depend upon their surroundings can be summarized as follows:

\begin{enumerate}
\item If $i$ has no neighbors employing $a(i)$ (as in $a(i) \notin \{ \sigma _{i-1} , \sigma _{i+1} \}$), then $i$'s best response is $a(i)$.

\item If $i$ has one neighbor playing $a(i)$ whilst its other neighbor is \emph{not} playing $b(i)$, then $i$'s best response is $b(i)$.

\item If $i$ has one neighbor playing $a(i)$ whilst its other neighbor plays $b(i)$, then $i$'s best response is $c(i)$.
\end{enumerate}

Player $i$'s best response will always be $a(i), b(i)$ or $c(i)$. Note that these values may be equal to one another.
Our method of proof is to show how to construct NE given the existence of players with various types.

\begin{lemma} \label{lem:loop1}
If there exists a player $i^*$ such that $a(i^*) = b(i^*)$, then there exists an NE.
\end{lemma}
\begin{proof}
Hold player $i^*$ fixed playing strategy $a(i^*)$ and allow the other players to evolve. In this scenario the remaining players $\mathcal{I} - \{ i^* \}$ are connected up in a line. The way these players evolve is described by a game with modified payoff functions $f_r ^i$.
These functions are defined so that $\forall i \in \mathcal{I} - \{ i^* \}$, $\forall r \in \mathcal{R}$, and $\forall x$, we have $f^i _r (x) = g^i _r (x)$, unless $r = a(i^*)$ and $i \in \{ i^* - 1 , i^* +1 \}$, in which case $f^i _r (x) = g^i _r (x+1)$.
Consider this modified game evolving upon the line graph induced upon the players $\mathcal{I} - \{ i^* \}$. Since the line is a tree graph, we can be assured (by theorem \ref{thm:tree}) that this system has an NE $\boldsymbol\sigma_{ - i^*}$.

Now let us reconsider player $i^*$ in the original system on the loop with payoff functions $g^i _r$. Suppose we set $\sigma _{i^*} = a(i^*)$ and allow the other players in $\mathcal{I} - \{ i^* \}$ to keep the strategies allocated to them under $\boldsymbol\sigma_{ - i^*}$.
Each player in $\mathcal{I} - \{ i^* \}$ will still be employing their best response in this configuration (because the modified system within which they reached this setup was essentially the same as the original setup with $i^*$'s strategy held fixed).
\begin{itemize}
\item If $i^*$ is such that $a(i^*) = b(i^*) = c(i^*)$, then $a(i^*)$ is always $i^*$'s best response and $i^*$ will also be satisfied under the configuration $\boldsymbol\sigma$ (which hence must be an NE).
\item If $\sigma _{i^* - 1} \neq a(i^*)$ or $\sigma _{i^* + 1} \neq a(i^*)$, then again $i^*$ is playing their best response and  $\boldsymbol\sigma$ is an NE.
\item If $\sigma _{i^* - 1} = a(i^*)  = \sigma _{i^* + 1 } \neq b(i^*) $, then $i^*$'s best response will be to change to employ $b(i^*)$. When $i^*$ switches its strategy in this way it will not decrease they payoff of its neighbors $i^* -1$ and $i^* + 1$ (which are not employing $b(i^*)$), nor will $i^*$'s change increase the incentive for $i^* -1$ or $i^* + 1$ to change to use a different strategy. It follows that once $i^*$ has switched its strategy to $b(i^*)$ the system will be in NE.
\end{itemize}
\end{proof}

To deal with the remaining cases, we \rev{will use the} algorithm defined below. \rev{It} takes in a value $\alpha \in \{0,1 \}$ and returns a strategy allocation $\boldsymbol\sigma$ to the players.

\begin{algorithm}\label{algo} \ \\
 If $\alpha = 0$ then $\sigma _0 := a(0)$, otherwise $\sigma _0 := b(0)$.\\
For $i$ from $1$ to $N-1$ do\\
If $a(i) = \sigma _ {i-1}$ then $\sigma _ {i} := b(i)$, otherwise $\sigma _ {i} := a(i)$.\\
end do.\\
Return($\boldsymbol\sigma$)\\
End
\end{algorithm}

Let $A(\alpha) = \boldsymbol\sigma$ denote the strategy allocation produced by Algorithm \ref{algo}.

\begin{lemma} \label{lem:algoprop}
Suppose $a(i) \neq b(i), \forall i \in \mathcal{I}$,
in this case the strategy allocation $A(\alpha) = \boldsymbol\sigma$ produced by Algorithm \ref{algo} has the following properties;

\begin{enumerate}
\item $\forall i \in \{0 , 1 ,.., N-2 \}$, we have $\sigma _i \neq \sigma _{i+1}$.

\item {$\sigma_{i-1} \neq \sigma_i \neq \sigma_{i+1}$ implies that $i$ is employing its best response under $\boldsymbol\sigma$, except in the case where $\alpha = 1$ \emph{and} $i=0$.} 

\item $\forall i \in \{1 , 2 ,.., N-2 \}$, player $i$ is employing its best response.

\item If $\alpha = 0$ and $\sigma _{N-1} \neq \sigma _{0}$, then the system is in NE.

 \item If we \emph{change} the strategy of player $i \in \{ 2 , 3 ,..., N-1 \}$ to some $r \neq \sigma_{i-1}$, then player $i-1$ will still be playing its best response in the resulting configuration.
\end{enumerate}
\end{lemma}

\begin{proof}
(1) follows from our assumption that $a(i) \neq b(i), \forall i$.

To see (2), note that if $\sigma_i = a(i)$, then $i$ is employing the best response to its surroundings because it has no neighbors employing the same strategy.
If $\sigma_i = b(i)$, then $i>0$ (by our assumption) and the nature of Algorithm \ref{algo} implies $\sigma _{i-1} = a(i)$. This means $i$'s best response is not $a(i)$. Also, supposing $c(i) \neq b(i)$, we can see that $i$'s best response is \emph{not} $c(i)$. This is because if it were, this would imply that $i$ has a neighbor employing $b(i)$, which is nonsensical because $i$ itself is employing $b(i)$ and we are assuming $i$'s strategy is different its neighbors strategies. Now we have shown that $i$'s best response is neither $a(i)$ nor $c(i)$ (when $c(i) \neq b(i)$), it follows that $i$'s best response is $b(i)$, which is what it is playing.

(3) Follows directly from (1) and (2).

(4) Follows directly from (1) and (2).

To see (5), suppose $i \in \{ 2 , 3 ,..., N-1 \}$ switches to $r \neq \sigma_{i-1}$. By (3), we know $i-1$ was employing its best response \emph{before} this switch. If $\sigma _{i-1} = a(i-1)$, then this strategy will clearly remain $i-1$'s best response.

 If $\sigma _{i-1} = b(i-1)$, then $\sigma _{i-2} = a(i-1)$. This implies $i-1$'s best response is \emph{not} $a(i-1)$.
Also, supposing $c(i-1) \neq b(i-1)$, we can see that $i$'s best response is \emph{not} $c(i-1)$. This is  because if it were, this would imply that $i-1$ has a neighbor employing $b(i-1)$, which is nonsensical because we know that $i-1$'s neighbors are playing $a(i-1) \neq b(i-1)$ and $r \neq b(i-1)$ respectively. So we have shown that $i-1$'s best response, in this case, must be $b(i-1)$.
\end{proof}

\begin{lemma}\label{lem:loop2}
Suppose $a(i) \neq b(i), \forall i \in \mathcal{I}$.
If there exists an $i^*$ such that $a(i^*) \neq b(i^*) \neq c(i^*) \neq a(i^*)$, then there exists an NE.
\end{lemma}

\begin{proof}
Suppose such an $i^*$ exists, and without loss of generality suppose $i^* = N-1$ (we can do this by relabeling the players without effecting the essential dynamics). Consider the configuration
$A(0) = \boldsymbol\sigma$ produced by Algorithm \ref{algo}. If $\sigma_{N-1} \neq \sigma _0$, then the system is an NE by part (4) of lemma \ref{lem:algoprop}. Next suppose $\sigma_{N-1} = \sigma _0$.
\begin{itemize}
\item If $\{ \sigma _{N-2} , \sigma _{0} \} = \{ a(N-1), b(N-1) \}$, then change $N-1$'s strategy to it's best response, which is $c(N-1)$. This will not cause $N-2$ to change their strategy by part (5) of lemma \ref{lem:algoprop}. Also this will not cause $0$ to  change its strategy, because it is employing $a(0)$ and still has no neighbors using this strategy. The system is hence in NE.
\item If, on the other hand, $\{ \sigma _{N-2} , \sigma _{0} \} \neq \{ a(N-1), b(N-1) \}$ then we must deal with two possibilities.
\begin{itemize}
\item If $\sigma _{N-1} = a(N-1)$, then $N-1$ will change to its best response which is $b(N-1)$. This will necessarily be different to $\sigma _{N-2}$ because of our assumption that $\{ \sigma _{N-2} , \sigma _{0} \} \neq \{ a(N-1), b(N-1) \}$ together with the fact that $\sigma_0 = \sigma _{N-1} = a(N-1)$.
    {After $N-1$ has changed its strategy to $b(N-1)$ its neighbors will be satisfied. In particular, $0$ will be satisfied because now it has no neighbors employing its strategy, $a(0)$. Also, $N-2$ will remain satisfied, by part (5) of lemma \ref{lem:algoprop}. It follows that, once $N-1$ has changed its strategy to $b(N-1)$, the system will be in NE.}
\item The other possibility we must deal with is that $\sigma _{N-1} = b(N-1)$. This actually cannot occur, because it implies $\sigma_0 = b(N-1)$ and $\sigma _{N-2} = a(N-1)$, which contradicts $\{ \sigma _{N-2} , \sigma _{0} \} \neq \{ a(N-1), b(N-1) \}$.
\end{itemize}
\end{itemize}
\end{proof}

Let $Q(i)$ denote the set (not multi-set) $\{ a(i) , b(i) , c(i) \}$ of $i$'s best responses  {in different scenarios}.

\begin{lemma} \label{lem:loop3}
 Suppose that $\forall i \in \mathcal{I}$ we have that $a(i) \neq b(i)$ and $c(i) \in \{ a(i) , b(i) \}$. In this case there must exists an NE.
 \end{lemma}

\begin{proof}
For any game satisfying the restrictions here, {\rev{one} of the following} is true
\begin{enumerate}
\item $\exists i \in \mathcal{I}, j \in \{ i-1 , i+1 \}$ such that $a(i) \notin Q(j)$,

\item $\exists i \in \mathcal{I}, j \in \{ i-1 , i+1 \}$ such that $b(i) \notin Q(j)$,

\item The game is equivalent to a game with $| \mathcal{R}| = 2$ resources.
\end{enumerate}

We show this by contradiction.
Suppose that (1) and (2) are both false, then since $Q(i) = \{ a(i) , b(i) \}$ and $a(i) \neq b(i)$ $\forall i$, we will have that $Q(i) = Q(i+1)$, $\forall i$. This implies that the game can be emulated by a game with $\mathcal{R} = Q(0)$ (because no player ever has a best response outside of $Q(0)$). Since $|Q(0)| = 2$ we have shown (3). Moreover, (3) implies that we can use Theorem \ref{thm:fip-2color} to prove the existence of an NE in the system.

From now on assume that (3) is false. This means that (1) or (2) must hold.

\rev{Assume} that (1) holds. We suppose, without loss of generality, that $a(0) \notin Q(N-1)$ (we can do this by relabeling the players without effecting the essential dynamics of the system). In this case, the configuration $A(0) = \boldsymbol\sigma$ will be a \rev{NE} by part (4) of lemma \ref{lem:algoprop}.

\rev{Assume} (1) does not hold. This implies that (2) is true, so we shall suppose (once again without loss of generality) that $b(0) \notin Q(N-1)$. Consider the configuration $A(1) = \boldsymbol\sigma$ generated by Algorithm \ref{algo}. Under this configuration, each player in $\{1,2,.., N-1 \}$ will be playing their best response according to part (3) of lemma \ref{lem:algoprop}. Player $N-1$ will also be playing its best response by part (2) of lemma \ref{lem:algoprop}.

\rev{We now} show that player $0$ is also employing \rev{its} best response. If $c(0) \neq b(0)$, then $c(0)$ is not player $0$'s best response. The reason for this is that $\sigma _{N-1} \neq b(0)$ and $\sigma _{1} \neq b(0)$ (recall that any $i$ must have a neighbor employing $b(i)$ in order for $c(i) \neq b(i)$ to be their best response).

\rev{Next} we show that $a(0)$ is not player $0$'s best response.
First, since (1) is false it must be that $a(0) \in Q(1) = \{ a(1) , b(1) \}$.
 If $a(0) = a(1)$ then $\sigma _1 = a(1)$ and we are done (because $0$ has a neighbor playing $a(0)$ ).

Now suppose, alternatively, that $a(0) = b(1)$.
In this case, the fact that (1) is false implies $a(1) \in Q(0) = \{ a(0) , b(0) \}$. Since $a(0) = b(1)$ and $b(1) \neq a(1)$, this means $a(1) = b(0)$. In other words, Algorithm \ref{algo} cannot assign $\sigma _1 := a(1)$ and so $\sigma _1 = b(1)$ (which is equal to $a(0)$). Now, since $0$ has a neighbor playing $a(0)$, this (once again) cannot be $0$'s best response.

Hence we have shown that $0$'s best response cannot be $a(0)$ and cannot be $c(0)$ (when $c(0) \neq b(0)$) and so $0$ is playing its best response, $b(0)$, and the system is in NE.
\end{proof}

With these lemmas in place, we can prove the stated result as follows.
 If $\exists i^* \in \mathcal{I}$ such that $a(i^*) = b(i^*)$, then lemma \ref{lem:loop1} implies the existence of an NE. Now suppose that
$a(i) \neq b(i)$ $\forall i \in \mathcal{I}$. If $\exists i^* \in \mathcal{I}$ such that $c(i^*) \notin \{ a(i^*),  b(i^*) \}$, then lemma \ref{lem:loop2} implies the existence of an NE. Now suppose that $c(i) \in \{ a(i) , b(i) \}$ $\forall i \in \mathcal{I}$. Now lemma \ref{lem:loop3} implies the existence of an NE. \hfill \qed



\end{document}